\documentclass[10pt,conference]{IEEEtran}
\IEEEoverridecommandlockouts
\usepackage{authblk}
\usepackage{cite}
\usepackage{amsmath,amssymb,amsfonts,dsfont}
\usepackage{algorithm}
\usepackage[noend]{algorithmic}
\usepackage{graphicx}
\usepackage{textcomp}
\usepackage{xcolor}
\usepackage{bm}
\usepackage{booktabs}
\usepackage{multirow}
\usepackage{lipsum}
\usepackage{url}
\usepackage{amsthm}
\usepackage{subcaption}

\ifCLASSOPTIONcompsoc
\usepackage[caption=false, font=normalsize, labelfont=sf, textfont=sf]{subfig}
\else
\usepackage[caption=false, font=footnotesize]{subfig}

\newtheorem{theorem}{Theorem}

\newtheorem{lemma}{Lemma}

\newcommand{\mbbR}{\mathbb{R}} 

\newcommand{\bdr}{\boldsymbol{r}} 
\newcommand{\bdw}{\boldsymbol{w}} 
  
\newcommand{\bdu}{\boldsymbol{u}}

\newcommand{\bdalpha}{\boldsymbol{\alpha}} 
\newcommand{\bdbeta}{\boldsymbol{\beta}}

\newcommand{\bdzeta}{\boldsymbol{\zeta}}

\newcommand{\bdq}{\boldsymbol{q}}

\newcommand{\mcL}{\mathcal{L}}


\allowdisplaybreaks[4] 

\begin{document}

\title{An Expectation-Maximization Relaxed Method for Privacy Funnel 
\thanks{The first two authors contributed equally to this work and $\dag$ marked the corresponding author. This work was partially supported by National Key Research and Development Program of China (2018YFA0701603) and National Natural Science Foundation of China (12271289 and 62231022).}
}

\author[1]{Lingyi Chen}
\author[1]{Jiachuan Ye}
\author[1]{Shitong Wu}
\author[2$\dag$]{Huihui Wu}
\author[1]{Hao Wu}
\author[3]{Wenyi Zhang}
\affil[1]{Department of Mathematical Sciences, Tsinghua University, Beijing 100084, China}
\affil[2]{Yangtze Delta Region Institute (Huzhou),
\authorcr University of Electronic Science and Technology of China, Huzhou, Zhejiang, 313000, P.R. China.} 
\affil[3]{Department of Electronic Engineering and Information Science, 
\authorcr University of Science and Technology of China, Hefei, Anhui 230027, China 
\authorcr Email: huihui.wu@ieee.org}

\maketitle


\begin{abstract}
The privacy funnel (PF) gives a framework of privacy-preserving data release, where the goal is to release useful data while also limiting the exposure of associated sensitive information. This framework has garnered significant interest due to its broad applications in characterization of the privacy-utility tradeoff. Hence, there is a strong motivation to develop numerical methods with high precision and theoretical convergence guarantees. In this paper, we propose a novel relaxation variant based on Jensen's inequality of the objective function for the computation of the PF problem. This model is proved to be equivalent to the original in terms of optimal solutions and optimal values. Based on our proposed model, we develop an accurate algorithm which only involves closed-form iterations. The convergence of our algorithm is theoretically guaranteed through descent estimation and Pinsker's inequality. Numerical results demonstrate the effectiveness of our proposed algorithm. 
\end{abstract}


\section{Introduction}

An increasing amount of private user data is flowing into the network nowadays, probably collected by certain individuals or companies eventually for customizing personalized services or other purposes. Usually, such data contains private or sensitive information. 
%
%
%
%
%
Considering general content of private information and the task of system, the problem is reduced to learning private representations, \textit{i.e.}, representations that are informative of the data (utility) but not of the private information. 
Researchers have started to model and study privacy protection mechanisms, in order to develop privacy preserving technologies and characterize the privacy-utility tradeoff. 
A general framework of statistical inference from an information-theoretic perspective has been proposed in \cite{calmon2012privacy}.
Specifically, given a public variable $X\in\mathcal{X}$ we want to transmit, and a correlated variable $S\in\mathcal{S}$ we want to keep private, one needs to encode $X$ as a variable $Y\in\mathcal{Y}$, forming a Markov chain
\[S\longleftrightarrow X\longrightarrow Y.\] 
Our goal is to minimize the average cost gain by the adversary after observation, while keeping the distortion of privacy-preserving mapping under certain threshold. 
When the self-information cost and log-loss metric are introduced, the privacy funnel (PF) \cite{makhdoumi2014privacyfunnel} is formulated to find a privacy preserving mapping from $\mathcal{X}$ to $\mathcal{Y}$, such that it minimizes the average information leakage $I(S;Y)$ with the disclosure $I(X;Y)$ kept above a certain threshold. 
More precisely, given the joint distribution $P_{S,X}$, the PF problem pursues the above tradeoff 
by considering the following optimization problem
\begin{equation}
\min\limits_{P_{Y|X}}I(S;Y),\quad\text{s.t.}~I(X;Y)\geq R,\label{pf-problem}
\end{equation}
where $R\leq H(X)$ to ensure that the problem is feasible\cite{makhdoumi2014privacyfunnel}. 
Research related to the PF model has covered a variety of setups in information theory \cite{calmon2015perfectprivacy,du2017Principal,asoodeh2019estimation}, machine learning \cite{romanelli2019generating,tripathy2019adversarial,bertran2019adversarial,rodriguez2021variational} and other fields. 
Moreover, the PF problem can be viewed as a dual of the well-known information bottleneck (IB) problem \cite{makhdoumi2014privacyfunnel,bu2021sdp}, suggesting an intriguing connection between them. 
However, different from the IB problem benefiting from a variety of algorithms including the BA algorithm \cite{blahut1972computation,tishby2000information} and the recently proposed ABP algorithm \cite{chen2023srib}, the PF problem still lacks an adequately effective algorithm. 
In fact, the PF problem is inherently non-convex, and therefore developing its numerical algorithms is a generally challenging task. 
Several algorithms have been proposed to solve the PF problem, but it is difficult to ensure effectiveness and convergence guarantee simultaneously. 
The greedy algorithm \cite{makhdoumi2014privacyfunnel} and the submodularity-based algorithm \cite{ding2019submodularity} motivated by agglomerative clustering \cite{slonim1999agglomerativeib} have been proposed early, merging the alphabet of sanitized variable to construct the mapping. 
Although strict descent of the objective is ensured, it only descends to a local minimum, resulting in limited computational accuracy and efficiency due to greedy search. 
The semi-definite programming (SDP) framework \cite{bu2021sdp} has also been applied to the PF problem, yet solving the SDP proves to be a time-consuming endeavor. 
Besides, variational approaches have been proposed in \cite{razeghi2023club,razeghi2024dvpf}, where each parameter is learnt through gradient descent. 
They only provide approximate results and require intensive computations to obtain gradient for each iteration, severely limiting the computational efficiency. 
A recent work \cite{huang2022splitting} has proposed a unified framework to solve the IB and PF problems through the Douglas-Rachford splitting (DRS) method, ensuring locally linear rate of convergence. 
%
%
%
%
This approach involves solving complex subproblems via gradient descent, and exceedingly large penalty is required in large-scale scenarios to ensure convergence, leading to gradient explosion and numerical instability. 
Consequently, the approach has been primarily suitable for limited scale cases. 
Such difficulty has been tackled with variational inference in a subsequent work\cite{huang2024dca}, but it only deals with a surrogate bound as an approximation.

In order to address the aforementioned difficulty, we propose a novel approach to solve the PF model by computing its upper bound relaxation variant, which is derived under the inspiration of the E-step of the celebrating Expectation-Maximization (EM) algorithm \cite{gupta2011EM}. 
This algorithm, named as the Alternating Expectation Minimization (AEM) algorithm, is developed to minimize the Lagrangian in an alternative manner, where each primal variable can be computed by a closed-form expression, with dual variables similarly updated or searched via Newton's method in only a few inner iterations. 
The closed-form solutions ensure the efficiency of our algorithm. 
Moreover, the descent of objective is theoretically estimated, and the convergence of iterative sequence to a Karush-Kuhn-Tucker (KKT) point is guaranteed by the Pinsker's inequality. 
Numerical experiments exhibit the effectiveness of our algorithm in a wide range of scenarios including traditional distributions and real-world datasets.  
%

\section{Problem Formulation}

Consider a discrete public variable $X\in\mathcal{X}$ with a relevant private variable $S\in\mathcal{S}$, and a representation variable $Y\in\mathcal{Y}$, where $\mathcal{S}=\{s_1,\cdots,s_K\}$, $\mathcal{X}=\{x_1,\cdots,x_M\}$, and $\mathcal{Y}=\{y_1,\cdots,y_N\}$. Then the Markov chain $S\leftrightarrow X\rightarrow Y$ yields the following distribution 
\begin{align*}
P_{S,Y}(s,y)=\sum\limits_{s,x,y}P_{S,X,Y}(s,x,y)=\sum\limits_x P_{S|X}(s|x)P_{X,Y}(x,y).
\end{align*}
Denote $s_{ki}=P_{S|X}(s_k|x_i)$, $u_{ij}=P_{X,Y}(x_i,y_j)$, $w_{ij}=P_{X|Y}(x_i|y_j)$, $r_j=P_Y(y_j)$, $p_i=P_X(x_i)$ and $\hat{R}=R+\sum\limits_i p_i\log p_i$, 
then the PF problem \eqref{pf-problem} can be formulated properly as the following constrained optimization problem
\begin{subequations}\label{original-pf}
\begin{align}
\min\limits_{\bdu,\bdw,\bdr}&\quad\sum\limits_{j,k}\Big(\sum\limits_i s_{ki}u_{ij}\Big)\Big(\log\big(\sum\limits_i s_{ki}u_{ij}\big)-\log r_j\Big), \\
\text{s.t.}&\quad\sum\limits_j u_{ij}=p_i,~\forall i;\quad\quad\sum\limits_i u_{ij}=r_j,~\forall j; \\
&\quad\sum\limits_i w_{ij}=1,~\forall j;\quad\quad u_{ij}=w_{ij}r_j,~\forall i,j; \\
&\quad\sum\limits_j r_j=1;\quad\quad\sum\limits_{i,j}u_{ij}\log w_{ij}\geq\hat{R},
\end{align}
\end{subequations}
where a feasible solution $\mathcal{Y}\supseteq\mathcal{X},p(x|y)=\textbf{1}\{x=y\}$ exists if $R\leq H(X)$\cite{makhdoumi2014privacyfunnel}. 
It is worth mentioning that our formulation \eqref{original-pf} guarantees the convexity with respect to each variable, which ensures numerical stability during optimization. 

\section{Upper Bound Relaxation Variant and Its Equivalence with the Original Model}
For the PF problem \eqref{original-pf} formulated above, it is still difficult to solve the variable $\bdu$, mainly due to the term taking the logarithm of the sum in the objective
\[f(\bdu,\bdr)=\sum\limits_{j,k}\Big(\sum\limits_i s_{ki}u_{ij}\Big)\Big(\log\big(\sum\limits_i s_{ki}u_{ij}\big)-\log r_j\Big).\]
Therefore, we propose the following upper bound to relax the problem (for similar techniques, see \cite{rodriguez2021variational,kolchinsky2019nonlinear}): 
\begin{equation}
\tilde{f}(\bdu,\bdr,\bdq)=\sum\limits_{i,j,k}s_{ki}u_{ij}\Big(\log(s_{ki}u_{ij})-\log r_j-\log q_{ijk}\Big), \label{objective}
\end{equation}
where $\sum\limits_i q_{ijk}=1$, making it much easier to minimize with respect to $\bdu$ from the first-order conditions. 
The corresponding novel model is the basis for us to design an alternating algorithm.

Surprisingly, as will be shown, the upper bound relaxation variant is equivalent to the original problem \eqref{original-pf}, so there is no loss considering such a variant. 

The following two subsections will specifically elaborate on these two points, constituting our main contributions. 
\subsection{Upper Bound Relaxation Variant}
First, we notice that the sum $\sum\limits_is_{ki}u_{ij}$ corresponds to $P_{S,Y}$, so the idea similar to the E-step of the EM algorithm \cite{gupta2011EM} inspires us to estimate $P_{X|S,Y}$ first, from which an upper bound of objective is established. 
In this scenario, $X$ corresponds to the latent variable in the EM algorithm. 
More specifically, the upper bound is given by
\begin{align*}
&\sum\limits_{j,k}\Big(\sum\limits_i s_{ki}u_{ij}\Big)\log\Big(\sum\limits_i s_{ki}u_{ij}\Big) \\
\leq&\sum\limits_{i,j,k} s_{ki}u_{ij} \Big(\log(s_{ki}u_{ij})-\log q_{ijk}\Big),
\end{align*}
where $\bdq$ is an auxiliary variable such that $\sum\limits_i q_{ijk}=1$. The equality holds if and only if $q_{ijk}=s_{ki}u_{ij}\Big/\Big(\sum\limits_{i'}s_{ki'}u_{i'j}\Big)$. 

Next, the following constraints \eqref{relaxed-constraints} given by the Markov chain and the transition probability conditions can be relaxed from our model, similar to the treatment in our previous work \cite{chen2023srib}. 
It is reasonable since they can be restored in our update scheme. Moreover, after relaxation we obtain an optimization problem that is convex with respect to each variable. It can be solved by analyzing the Lagrangian with closed-form iterations. 
\begin{subequations}\label{relaxed-constraints}
\begin{align}
&\sum\limits_i u_{ij}=r_j,\quad\forall j;\label{relaxed-r} \\
&u_{ij}=w_{ij}r_j,\quad\forall i,j;\label{relaxed-w} \\
&q_{ijk}=s_{ki}u_{ij}\Big/\Big(\sum\limits_{i'}s_{ki'}u_{i'j}\Big),\quad\forall i,j,k.\label{relaxed-q}
\end{align}
\end{subequations}

Under these crucial observations, the process of directly optimizing $\bdu$ in the original PF problem \eqref{original-pf} can be transformed to that of estimating $\bdq$ first, and optimizing $\bdu$ thereafter in our upper bound relaxation variant: 
\begin{subequations}\label{relaxed-pf}
\begin{align}
\min\limits_{\bdu,\bdw,\bdr,\bdq}&\quad\sum\limits_{i,j,k} s_{ki}u_{ij}\Big(\log(s_{ki}u_{ij})-\log r_j-\log q_{ijk}\Big), \\
\text{s.t.}&\quad\sum\limits_j u_{ij}=p_i,~\forall i;\quad\quad\sum\limits_i w_{ij}=1,~\forall j; \label{constraint-w}\\
&\quad\sum\limits_j r_j=1;\quad\quad\quad\quad\sum\limits_i q_{ijk}=1,~\forall j,k; \label{constraint-rq}\\
&\quad\sum\limits_{i,j}u_{ij}\log w_{ij}\geq\hat{R}.
\end{align}
\end{subequations}

\subsection{Equivalence of Optimal Solutions}

We establish the equivalence between model \eqref{original-pf} and its upper bound relaxation variant \eqref{relaxed-pf} as shown in the following theorem.
\begin{theorem}
The optimal values as well as the optimal triples $(\bdu^\star,\bdw^\star,\bdr^\star)$ of \eqref{original-pf} and \eqref{relaxed-pf} are identical. 
\end{theorem}
\begin{proof}
Suppose $(\bdu^\star,\bdw^\star,\bdr^\star)$ is optimal for \eqref{original-pf}, then $(\bdu^\star,\bdw^\star,\bdr^\star,\bdq^\star)$ is feasible for \eqref{relaxed-pf} if we define $q_{ijk}^\star=s_{ki}u_{ij}^\star\Big/\Big(\sum\limits_{i'}s_{ki'}u_{i'j}^\star\Big)$. 
The expression of $\bdq^\star$ implies the identity of two objectives. Since \eqref{relaxed-pf} is an upper bound of \eqref{original-pf}, $(\bdu^\star,\bdw^\star,\bdr^\star,\bdq^\star)$ is optimal for \eqref{relaxed-pf}. 

On the other hand, suppose $(\bdu^\star,\bdw^\star,\bdr^\star,\bdq^\star)$ is optimal for \eqref{relaxed-pf}, then the KKT conditions yield the existence of $\gamma^\star\in\mathbb{R}$, $\bdbeta^\star\in\mathbb{R}^N$, $\bdzeta\in\mathbb{R}^{N\times K}$ such that
\begin{subequations}\label{opt-conditions}
\begin{align}
&\gamma^\star-\sum\limits_{i,k}s_{ki}u_{ij}^\star\Big/r_j^\star=0,~r_j^\star=\Big(\sum\limits_iu_{ij}^\star\Big)\Big/\gamma^\star; \label{opt-r}\\
&\beta_j^\star-\lambda^\star u_{ij}^\star/w_{ij}^\star=0,~w_{ij}^\star=\lambda^\star u_{ij}^\star/\beta_j^\star; \label{opt-w} \\
&\zeta_{jk}^\star-s_{ki}u_{ij}^\star/q_{ijk}^\star=0,~q_{ijk}^\star=s_{ki}u_{ij}^\star/\zeta_{jk}^\star.\label{opt-q}
\end{align}
\end{subequations}
Substituting \eqref{opt-conditions} into \eqref{relaxed-pf} we get $\gamma^\star=1$, $\beta_j^\star=\lambda^\star\sum\limits_iu_{ij}^\star=\lambda^\star r_j^\star$, $\zeta_{jk}^\star=\sum\limits_i s_{ki}u_{ij}^\star$, and thus \eqref{relaxed-constraints} is satisfied and $(\bdu^\star,\bdw^\star,\bdr^\star)$ is feasible for \eqref{original-pf}. Suppose it is not optimal for \eqref{original-pf}, then there exists $(\bdu,\bdw,\bdr)$ such that 
\begin{equation*}
\tilde{f}(\bdu,\bdr,\bdq)=f(\bdu,\bdr)<f(\bdu^\star,\bdr^\star)=\tilde{f}(\bdu^\star,\bdr^\star,\bdq^\star),
\end{equation*}
where $q_{ijk}=s_{ki}u_{ij}\Big/\Big(\sum\limits_{i'}s_{ki'}u_{i'j}\Big)$, and the second equality follows from the expression of $\bdq^\star$. Then $(\bdu^\star,\bdw^\star,\bdr^\star,\bdq^\star)$ is not optimal for \eqref{relaxed-pf}, which is a contradiction.  
It means that $(\bdu^\star,\bdw^\star,\bdr^\star)$ is optimal for \eqref{original-pf}. 
\end{proof}


\section{The Alternating Expectation Minimization Algorithm}
In this section, we propose a convergence guaranteed alternating algorithm to solve problem \eqref{relaxed-pf}. 
Since the update of $\bdq$ corresponds to the E-step of the EM algorithm, and the update of other variables corresponds to the M-step, we name it the Alternating Expectation Minimization (AEM) algorithm. 
\subsection{Algorithm Derivation and Implementation}
We introduce multipliers $\bdalpha\in\mbbR^M$, $\bdbeta\in\mbbR^N$, $\gamma\in\mbbR$, $\bdzeta\in\mbbR^{N\times K}$, $\lambda\in\mbbR^+$ and obtain the Lagrangian of \eqref{relaxed-pf}: 
\begin{align*}
&\mathcal{L}(\bdu,\bdw,\bdr,\bdq;\bdalpha,\bdbeta,\gamma,\bdzeta,\lambda)=\sum\limits_{i,j,k}s_{ki}u_{ij}\Big(\log(s_{ki}u_{ij})-\log r_j \\
&-\log q_{ijk}\Big)+\sum\limits_i\alpha_i\Big(\sum\limits_j u_{ij}-p_i\Big)+\sum\limits_j\beta_j\Big(\sum\limits_i w_{ij}-1\Big) \\
&+\gamma\Big(\sum\limits_j r_j-1\Big)+\sum\limits_{j,k}\zeta_{jk}\Big(\sum\limits_i q_{ijk}-1\Big) \\
&-\lambda\Big(\sum\limits_{i,j}u_{ij}\log w_{ij}-\hat{R}\Big).
\end{align*}
Our key ingredient is to alternatively update the primal variables with their corresponding dual variables simultaneously updated. 
Based on the convexity of the Lagrangian with respect to each variable, we take partial derivatives for each primal variable and obtain their closed-form iterative expressions. 
This update scheme ensures high computation efficiency and offers an accurate descent estimation of the objective.

\subsubsection{Updating $\bdq$ and $\bdzeta$}
The first-order condition yields
\begin{equation*}
\dfrac{\partial\mcL}{\partial q_{ijk}}=-\dfrac{s_{ki}u_{ij}}{q_{ijk}}+\zeta_{jk}=0,~q_{ijk}=\dfrac{s_{ki}u_{ij}}{\zeta_{jk}}.
\end{equation*}
Substituting them into the constraint of $\bdq$ we have
\begin{equation*}
\sum\limits_i(s_{ki}u_{ij})\Big/\zeta_{jk}=1,
~\zeta_{jk}=\sum\limits_i s_{ki}u_{ij}.
\end{equation*}
Then we can update $\bdq$ by
\begin{equation*}
q_{ijk}=s_{ki}u_{ij}\Big/\Big(\sum\limits_{i'}s_{ki'}u_{i'j}\Big),
\end{equation*}
which is exactly the relaxed constraint of $\bdq$ in \eqref{relaxed-q}. 

\subsubsection{Updating $\bdu$ and $\bdalpha,\lambda$}
Define $\phi_{ij}=\sum\limits_k s_{ki}\big(\log q_{ijk}-\log s_{ki}\big)$, then the first-order condition yields
\begin{subequations}
\begin{align*}
\dfrac{\partial\mcL}{\partial u_{ij}}&=\log u_{ij}-\phi_{ij}+1+\alpha_i-\log r_j-\lambda\log w_{ij}=0, \\
u_{ij}&=e^{\lambda\log w_{ij}+\phi_{ij}-\alpha_i-1}r_j.
\end{align*}
\end{subequations}
Substituting them into the constraint of $\bdu$ we have
\begin{subequations}
\begin{align*}
&\sum\limits_j e^{\lambda\log w_{ij}+\phi_{ij}-\alpha_i-1}r_j=p_i, \\
\alpha_i&=\log\Big(\sum\limits_j e^{\lambda\log w_{ij}+\phi_{ij}}r_j\Big)-\log p_i-1.
\end{align*}
\end{subequations}
Then we can update $\bdu$ by
\begin{align*}
u_{ij}=\dfrac{e^{\lambda\log w_{ij}+\phi_{ij}}r_j}{\sum\limits_{j'}e^{\lambda\log w_{ij'}+\phi_{ij'}}r_{j'}}p_i,
\end{align*}
where $\lambda$ can be updated by finding the unique root of the following monotonic function\footnote{We can easily verify that $G(\lambda)\geq 0$ as discussed in \cite{ye2022lmrate,wu2023commot}.} via the Newton's method: 
\begin{subequations}
\begin{align*}
G(\lambda)&=\sum\limits_{i,j}\dfrac{e^{\lambda\log w_{ij}+\phi_{ij}}r_j}{\sum\limits_{j'}e^{\lambda\log w_{ij'}+\phi_{ij'}}r_{j'}}p_i\log w_{ij}-\hat{R}=0.
\end{align*}
\end{subequations}

\subsubsection{Updating $\bdr$ and $\gamma$}
The first-order condition yields
\begin{equation*}
\dfrac{\partial\mcL}{\partial r_j}=-\sum\limits_{i,k}\dfrac{s_{ki}u_{ij}}{r_j}+\gamma=0,~r_j=\sum\limits_i\dfrac{u_{ij}}{\gamma}.
\end{equation*}
Substituting them into the constraint of $\bdr$ we have
\begin{equation*}
\sum\limits_{i,j}u_{ij}\Big/\gamma=1,~\gamma=1,
\end{equation*}
which follows from the fact that $\sum\limits_{i,j}u_{ij}=1$. Then we can update $\bdr$ by
\begin{equation*}
r_j=\sum\limits_i u_{ij},
\end{equation*}
which is exactly the relaxed constraint of $\bdr$ in \eqref{relaxed-r}. 

\subsubsection{Updating $\bdw$ and $\bdbeta$}
The first-order condition yields
\begin{equation*}
\dfrac{\partial\mcL}{\partial w_{ij}}=-\dfrac{\lambda u_{ij}}{w_{ij}}+\beta_j=0,~w_{ij}=\dfrac{\lambda u_{ij}}{\beta_j}.
\end{equation*}
Substituting them the constraint of $\bdw$ we have
\begin{equation*}
\sum\limits_i\lambda u_{ij}\Big/\beta_j=1,~\beta_j=\lambda\sum\limits_i u_{ij}.
\end{equation*}
Then we can update $\bdw$ by
\begin{equation*}
w_{ij}=u_{ij}\Big/\Big(\sum\limits_{i'}u_{i'j}\Big)=u_{ij}/r_j,
\end{equation*}
which is exactly the relaxed constraint of $\bdw$ in \eqref{relaxed-w}. 

To summarize, the proposed AEM algorithm is presented in Algorithm \ref{alg-pf}. 

\begin{algorithm}
\caption{Alternating Expectation Minimization (AEM)}
\label{alg-pf}
\begin{algorithmic}[ht]
\STATE{\bf Input} $p_i=p(x_i)$, $s_{ki}=p(s_k|x_i)$, $\hat{R}$, \textit{max\_iter}
\STATE{\bf Output} min $\sum\limits_{i,j,k}s_{ki}u_{ij}\Big(\log(s_{ki}u_{ij})-\log r_j-\log q_{ijk}\Big)$
\STATE Initialize a feasible solution $u_{ij}=\frac{\textbf{1}\{i=j\}}{M},r_j=\sum\limits_i u_{ij}$
\FOR{$n=1:\textit{max\_iter}$}
    \STATE $q_{ijk}\leftarrow s_{ki}u_{ij}\Big/\Big(\sum\limits_{i'}s_{ki'}u_{i'j}\Big)$
    \STATE $\phi_{ij}\leftarrow\sum\limits_k s_{ki}\big(\log q_{ijk}-\log s_{ki}\big)$
    \STATE Find $\lambda$ such that $G(\lambda)=0$ 
    \STATE $u_{ij}\leftarrow\dfrac{e^{\lambda\log w_{ij}+\phi_{ij}}r_j}{\sum\limits_{j'}e^{\lambda\log w_{ij'}+\phi_{ij'}}r_{j'}}p_i$
    \STATE $r_j\leftarrow\sum\limits_i u_{ij}$
    \STATE $w_{ij}\leftarrow u_{ij}/r_j$
\ENDFOR
\STATE{\bf Return} $\sum\limits_{i,j,k}s_{ki}u_{ij}\Big(\log(s_{ki}u_{ij})-\log r_j-\log q_{ijk}\Big)$
\end{algorithmic}
\end{algorithm}


\subsection{Convergence Analysis}

With the guarantee of Theorem 1, we estimate the descent of the objective \eqref{objective} for model \eqref{relaxed-pf}. For short, the closed-form updates of each primal variable provide the descent in the form of Kullback-Leibler (KL) divergence between the corresponding variables in two consecutive iterations. 

\begin{lemma}
The objective $\tilde{f}(\bdw,\bdr,\bdq)$ is non-increasing, i.e.
\begin{align*}
&\tilde{f}(\bdu^{(n+1)},\bdr^{(n+1)},\bdq^{(n+1)})\leq \tilde{f}(\bdu^{(n+1)},\bdr^{(n)},\bdq^{(n+1)}) \\ \leq&\tilde{f}(\bdu^{(n)},\bdr^{(n)},\bdq^{(n+1)})\leq \tilde{f}(\bdu^{(n)},\bdr^{(n)},\bdq^{(n)}).
\end{align*}
Moreover, the descent of objective can be estimated by
\begin{multline}\label{descent-estimation}
\tilde{f}(\bdu^{(n)},\bdr^{(n)},\bdq^{(n)})-\tilde{f}(\bdu^{(n+1)},\bdr^{(n+1)},\bdq^{(n+1)}) \\
=\sum\limits_{j,k}\Big(\sum\limits_i s_{ki}u_{ij}^{(n)}\Big)D(\bdq_{jk}^{(n+1)}\Vert\bdq_{jk}^{(n)})+D(\bdr^{(n+1)}\Vert\bdr^{(n)}) \\
+D(\bdu^{(n)}\Vert\bdu^{(n+1)})+\lambda^{(n+1)}\sum\limits_j r_j^{(n)}D(\bdw_j^{(n)}\Vert\bdw_j^{(n-1)}),
\end{multline}
where $\bdq_{jk}$ denotes the row of $\bdq$ where the $j$-th column slice intersects with the $k$-th vertical slice, and $\bdw_j$ denotes the $j$-th column of $\bdw$. 
\end{lemma}
\begin{proof}
We have the following estimations proved in Appendix A: 
\begin{align*}
&\tilde{f}(\bdu^{(n)},\bdr^{(n)},\bdq^{(n)})-\tilde{f}(\bdu^{(n)},\bdr^{(n)},\bdq^{(n+1)}) \\
=&\sum\limits_{j,k}\Big(\sum\limits_i s_{ki}u_{ij}^{(n)}\Big)D(\bdq_{jk}^{(n+1)}\Vert\bdq_{jk}^{(n)}), \\
&\tilde{f}(\bdu^{(n)},\bdr^{(n)},\bdq^{(n+1)})-\tilde{f}(\bdu^{(n+1)},\bdr^{(n)},\bdq^{(n+1)}) \\
=&D(\bdu^{(n)}\Vert\bdu^{(n+1)})+\lambda^{(n+1)}\sum\limits_j r_j^{(n)}D(\bdw_j^{(n)}\Vert\bdw_j^{(n-1)}), \\
&\tilde{f}(\bdu^{(n+1)},\bdr^{(n)},\bdq^{(n+1)})-f(\bdu^{(n+1)},\bdr^{(n+1)},\bdq^{(n+1)}) \\
=&D(\bdr^{(n+1)}\Vert\bdr^{(n)}). \\
\end{align*}
The non-increasing property follows from the non-negativity of KL divergence. 
\end{proof}

\begin{lemma}
The objective $\tilde{f}(\bdu,\bdr,\bdq)$ is non-negative.  
\end{lemma}
\begin{proof}
\begin{align*}
\tilde{f}(\bdu,\bdr,\bdq)=&\sum\limits_{i,j,k}s_{ki}u_{ij}\Big(\log(s_{ki}u_{ij})-\log r_j-\log q_{ijk}\Big) \\
\geq&\sum\limits_{i,j,k}s_{ki}u_{ij}\log(s_{ki}u_{ij}) \\\geq&\Big(\sum\limits_{i,j,k}s_{ki}u_{ij}\Big)\log\Big(\sum\limits_{i,j,k}s_{ki}u_{ij}\Big)=0.
\end{align*}
The first inequality is due to $0\leq r_j\leq 1$, $0\leq q_{ijk}\leq 1$, and the second inequality follows from Jensen's inequality. 
\end{proof}

The objective converges since it is non-increasing and lower bounded throughout iterations. Furthermore, the convergence of iterative sequence is also guaranteed.

\begin{theorem}
The sequence $\{(\bdu^{(n)},\bdw^{(n)},\bdr^{(n)})\}$ converges to a stationary point $(\bdu^\star,\bdw^\star,\bdr^\star)$. 
\end{theorem}
\begin{proof}
Applying Pinsker's inequality to \eqref{descent-estimation} we have 
\begin{align*}
&\tilde{f}(\bdu^{(n)},\bdr^{(n)},\bdq^{(n)})-\tilde{f}(\bdu^{(n+1)},\bdr^{(n+1)},\bdq^{(n+1)}) \\
\geq&\dfrac{1}{2}\Bigg(\sum\limits_{j,k}\Big(\sum\limits_i s_{ki}u_{ij}^{(n)}\Big)\Vert\bdq_{jk}^{(n+1)}-\bdq_{jk}^{(n)}\Vert^2+\Vert\bdr^{(n+1)}-\bdr^{(n)}\Vert_1^2 \\
+&\Vert\bdu^{(n)}-\bdu^{(n+1)}\Vert_1^2+\lambda^{(n+1)}\sum\limits_j r_j^{(n)}\Vert\bdw_j^{(n)}-\bdw_j^{(n-1)}\Vert_1^2\Bigg)\geq 0.
\end{align*}
Since the objective converges, we have
\begin{equation*}
\sum\limits_{n=1}^\infty\big(\tilde{f}(\bdu^{(n)},\bdr^{(n)},\bdq^{(n)})-\tilde{f}(\bdu^{(n+1)},\bdr^{(n+1)},\bdq^{(n+1)})\big)<+\infty.
\end{equation*}
This implies $\sum\limits_{n=1}^\infty\Vert\bdr^{(n+1)}-\bdr^{(n)}\Vert_1^2<+\infty$, so $\{\bdr^{(n)}\}$ converges. Similar analysis goes for $\{\bdu^{(n)}\}$. The update rule of $\bdw$ ensures the convergence of $\{\bdw^{(n)}\}$. 
Denote the limit point by $(\bdu^\star,\bdw^\star,\bdr^\star)$, then the iterative scheme guarantees the feasibility of the limit point. 
Taking limit of our iterative expressions, the KKT conditions are satisfied at $(\bdu^\star,\bdw^\star,\bdr^\star)$.
\end{proof}
%


\section{Numerical Results}

This section evaluates our AEM algorithm on a synthetic distribution and two real-world datasets of different sizes. 
These experiments have been implemented by Matlab R2023b on a laptop with 16G RAM and one Intel(R) Core(TM) i7-12700H CPU @ 2.30GHz. 

\subsection{Experiments on A Synthetic Distribution}
The synthetic conditional distribution in our experiment is given by\cite{huang2022splitting}
\begin{align*}
P_{S|X}&=\begin{pmatrix} 0.9 & 0.08 & 0.4 \\ 0.025 & 0.82 & 0.05 \\ 0.075 & 0.1 & 0.55\end{pmatrix}.
\end{align*}
We evaluate the performance with uniform and non-uniform $P_X$ respectively given by
\begin{align*}
P_{X,\text{unif}}=\begin{pmatrix}1/3 & 1/3 & 1/3\end{pmatrix}^T,P_{X,\text{nonunif}}=\begin{pmatrix}0.1 & 0.3 & 0.6\end{pmatrix}^T.
\end{align*}
In our experiment, we set $N=4$ and the maximum number of iterations 500. 
We compare the AEM algorithm with the DRS method \cite{huang2022splitting} and plot PF curves on the information plane given in Fig. \ref{fig:synthetic}. 
The reported values $I(S;Y)$ are the best ones by performing 30 different trials for each $R\leq H(X)$.

\begin{figure}[htbp]
    \centering
    \includegraphics[width=\linewidth]{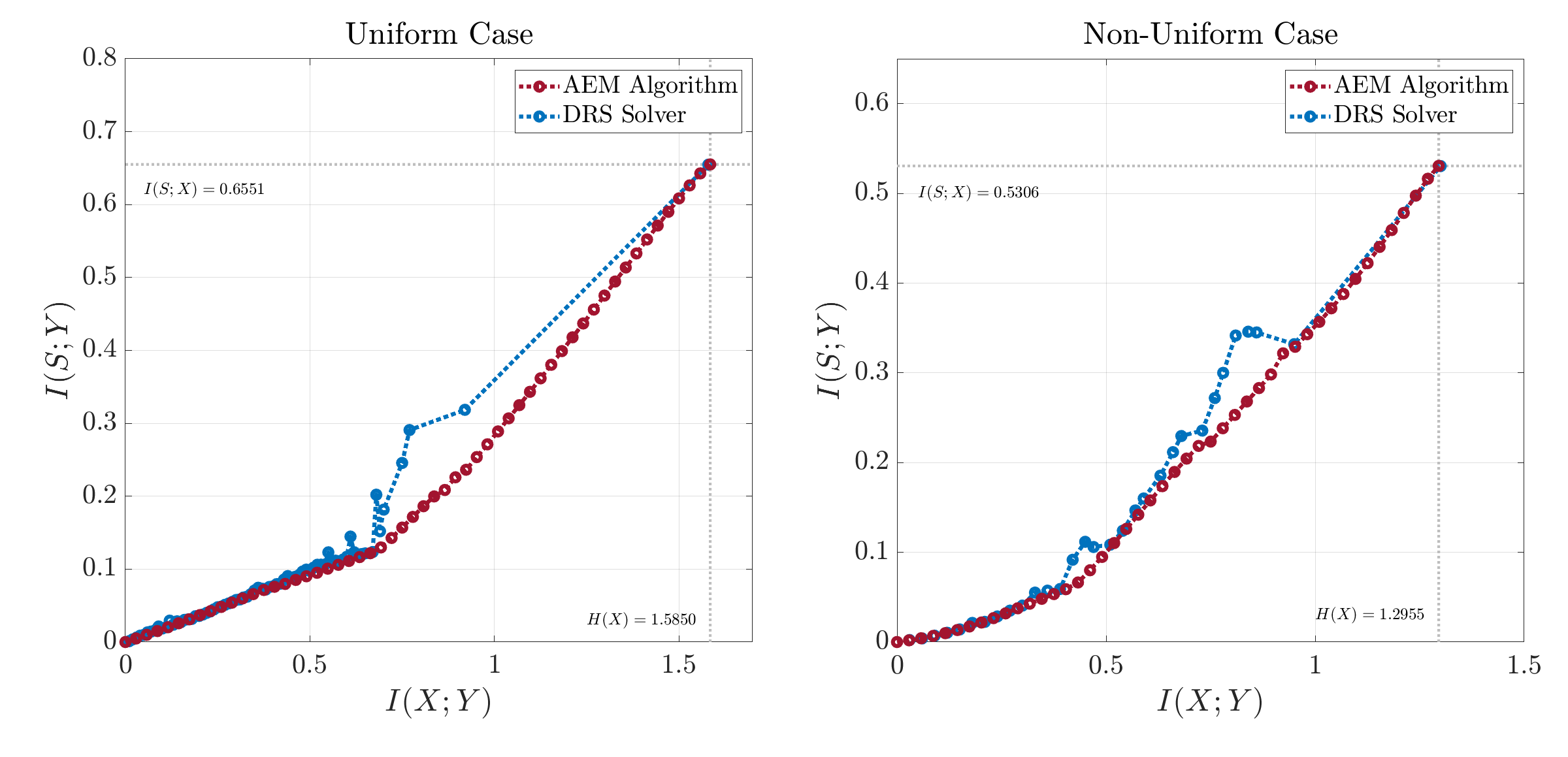}
    \caption{Comparison of PF curves between the AEM algorithm (red dashed line) and the DRS method (blue dashed line).}
    \label{fig:synthetic}
\end{figure}

Both algorithms reach the theoretical bound $I(X;Y)=H(X)$ and $I(S;Y)=I(S;X)$. 
Compared to the DRS method, the proposed AEM algorithm provides more uniform points and smoother curves, demonstrating its numerical stability. 
In contrast, the DRS method performs worse on certain disclosures and shows poorer ability to output complete curves.

\subsection{Experiment on Real-World Datasets}
We evaluate the performance on the following two datasets: ``Heart failure clinical records" dataset and ``Census income" dataset from the UCI Machine Learning Repository\cite{blake1998uci}. 
The former dataset has 299 items and 13 attributes. We select $\mathcal{S}=\{$``sex",``death"\} and $\mathcal{X}=\{$``anaemia",``high blood pressure",``diabetes",``smoking"$\}$ where all selected attributes are binary, so $|\mathcal{S}|=4$, $|\mathcal{X}|=16$. 
The latter dataset has 32561 items and 14 attributes. We select $\mathcal{S}=\{$``age",``income level"$\}$, and $\mathcal{X}=\{$``age",``gender",``education level"$\}$ where all selected attributes are all integers, so $|\mathcal{S}|=10$, $|\mathcal{X}|=160$. 
The distributions are taken empirically and normalized after adding a perturbation of $10^{-3}$ to each entry. 

\begin{figure}[htbp]
    \centering
    \includegraphics[width=\linewidth]{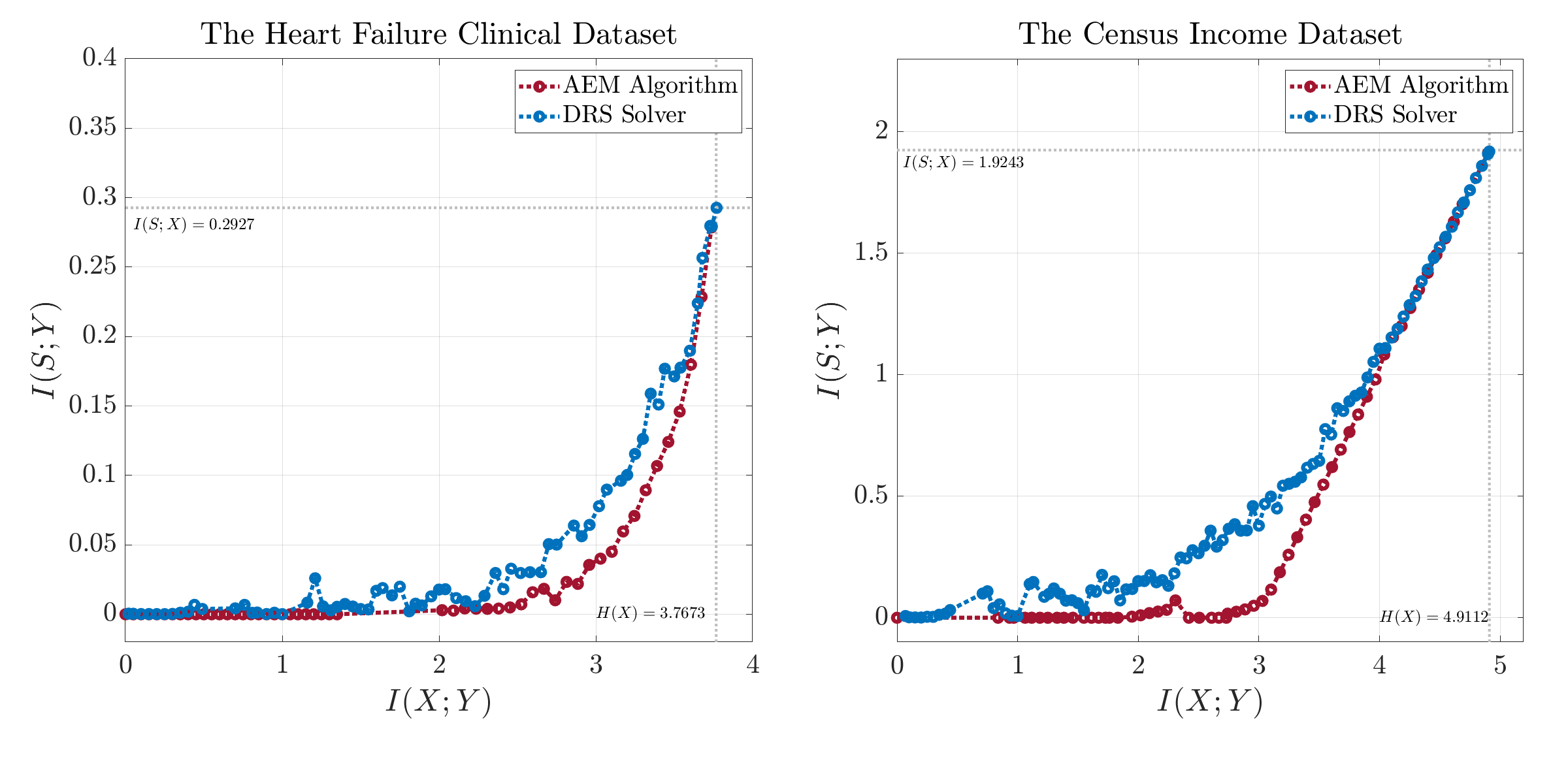}
    \caption{Performance comparison between the AEM algorithm and the DRS method on two real-world datasets. }
    \label{fig:clinical}
\end{figure}

As shown in Fig. \ref{fig:clinical}, the AEM algorithm reaches almost perfect privacy (\textit{i.e.}, $I(S;Y)\approx 0$) for a wide range of disclosure thresholds. 
In contrast, we found that the DRS method performs worse in a range of disclosures. This may be due to the difficulty of striking a balance between convergence and numerical stability, as discussed in Section I. A specific explanation is attached in Appendix B. 
%
%

\section{Conclusion}
We propose a novel approach to solve the PF problem where the objective is replaced with an upper bound under the EM framework, and several constraints given by the Markov chain and transition probability conditions are relaxed. 
The equivalence between the original model and the upper bound relaxation variant is further proven. 
Based on the new model, we develop the AEM algorithm by analyzing the Lagrangian, which turns out to recover the relaxed constraints with theoretically guaranteed convergence. 
Numerical experiments on synthetic and real-world datasets demonstrate effectiveness of our approach. 
The extension of our approach to continuous scenarios is also interesting and worthy of further research. 

\newpage
\bibliographystyle{IEEEtran}
\bibliography{ref}

\begin{thebibliography}{10}
\providecommand{\url}[1]{#1}
\csname url@samestyle\endcsname
\providecommand{\newblock}{\relax}
\providecommand{\bibinfo}[2]{#2}
\providecommand{\BIBentrySTDinterwordspacing}{\spaceskip=0pt\relax}
\providecommand{\BIBentryALTinterwordstretchfactor}{4}
\providecommand{\BIBentryALTinterwordspacing}{\spaceskip=\fontdimen2\font plus
\BIBentryALTinterwordstretchfactor\fontdimen3\font minus \fontdimen4\font\relax}
\providecommand{\BIBforeignlanguage}[2]{{%
\expandafter\ifx\csname l@#1\endcsname\relax
\typeout{** WARNING: IEEEtran.bst: No hyphenation pattern has been}%
\typeout{** loaded for the language `#1'. Using the pattern for}%
\typeout{** the default language instead.}%
\else
\language=\csname l@#1\endcsname
\fi
#2}}
\providecommand{\BIBdecl}{\relax}
\BIBdecl

\bibitem{calmon2012privacy}
F.~du~Pin~Calmon and N.~Fawaz, ``{Privacy Against Statistical Inference},'' in \emph{50th Annual Allerton Conference on Communication, Control, and Computing (Allerton)}, 2012, pp. 1401--1408.

\bibitem{makhdoumi2014privacyfunnel}
A.~Makhdoumi, S.~Salamatian, N.~Fawaz, and M.~M{\'e}dard, ``{From the Information Bottleneck to the Privacy Funnel},'' in \emph{2014 IEEE Information Theory Workshop (ITW)}, 2014, pp. 501--505.

\bibitem{calmon2015perfectprivacy}
F.~P. Calmon, A.~Makhdoumi, and M.~M{\'e}dard, ``{Fundamental Limits of Perfect Privacy},'' in \emph{2015 IEEE International Symposium on Information Theory (ISIT)}, 2015, pp. 1796--1800.

\bibitem{du2017Principal}
F.~du~Pin~Calmon, A.~Makhdoumi, M.~M{\'e}dard, M.~Varia, M.~Christiansen, and K.~R. Duffy, ``{Principal Inertia Components and Applications},'' \emph{IEEE Transactions on Information Theory}, vol.~63, no.~8, pp. 5011--5038, 2017.

\bibitem{asoodeh2019estimation}
S.~Asoodeh, M.~Diaz, F.~Alajaji, and T.~Linder, ``{Estimation Efficiency Under Privacy Constraints},'' \emph{IEEE Transactions on Information Theory}, vol.~65, no.~3, pp. 1512--1534, 2019.

\bibitem{romanelli2019generating}
M.~Romanelli, C.~Palamidessi, and K.~Chatzikokolakis, ``{Generating Optimal Privacy-Protection Mechanisms via Machine Learning},'' \emph{arXiv preprint arXiv:1904.01059}, 2019.

\bibitem{tripathy2019adversarial}
A.~Tripathy, Y.~Wang, and P.~Ishwar, ``{Privacy-Preserving Adversarial Networks},'' in \emph{2019 57th Annual Allerton Conference on Communication, Control, and Computing (Allerton)}, 2019, pp. 495--505.

\bibitem{bertran2019adversarial}
M.~Bertran, N.~Martinez, A.~Papadaki, Q.~Qiu, M.~Rodrigues, G.~Reeves, and G.~Sapiro, ``{Adversarially Learned Representations for Information Obfuscation and Inference},'' in \emph{36th International Conference on Machine Learning (ICML)}, vol.~97, 2019, pp. 614--623.

\bibitem{rodriguez2021variational}
B.~Rodr{\'\i}guez-G{\'a}lvez, R.~Thobaben, and M.~Skoglund, ``{A Variational Approach to Privacy and Fairness},'' in \emph{2021 IEEE Information Theory Workshop (ITW)}, 2021, pp. 1--6.

\bibitem{bu2021sdp}
Y.~Bu, T.~Wang, and G.~W. Wornell, ``{SDP Methods for Sensitivity-Constrained Privacy Funnel and Information Bottleneck Problems},'' in \emph{2021 IEEE International Symposium on Information Theory (ISIT)}, 2021, pp. 49--54.

\bibitem{blahut1972computation}
R.~Blahut, ``{Computation of Channel Capacity and Rate-Distortion Functions},'' \emph{IEEE Transactions on Information Theory}, vol.~18, no.~4, pp. 460--473, 1972.

\bibitem{tishby2000information}
N.~Tishby, F.~C. Pereira, and W.~Bialek, ``{The Information Bottleneck Method},'' in \emph{37th Annual Allerton Conference on Communications, Control and Computing (Allerton)}, 1999, pp. 368--377.

\bibitem{chen2023srib}
L.~Chen, S.~Wu, J.~Ye, H.~Wu, W.~Zhang, and H.~Wu, ``{Efficient and Provably Convergent Computation of Information Bottleneck: A Semi-Relaxed Approach},'' \emph{IEEE International Conference on Communications (ICC)}, 2024, accepted.

\bibitem{ding2019submodularity}
N.~Ding and P.~Sadeghi, ``{A Submodularity-Based Agglomerative Clustering Algorithm for the Privacy Funnel},'' \emph{arXiv preprint arXiv:1901.06629}, 2019.

\bibitem{slonim1999agglomerativeib}
N.~Slonim and N.~Tishby, ``{Agglomerative Information Bottleneck},'' in \emph{Advances in Neural Information Processing Systems (NIPS)}, vol.~12, 1999.

\bibitem{razeghi2023club}
B.~Razeghi, F.~P. Calmon, D.~Gunduz, and S.~Voloshynovskiy, ``{Bottlenecks CLUB: Unifying Information-Theoretic Trade-offs Among Complexity, Leakage, and Utility},'' \emph{IEEE Transactions on Information Forensics and Security}, vol.~18, pp. 2060--2075, 2023.

\bibitem{razeghi2024dvpf}
B.~Razeghi, P.~Rahimi, and S.~Marcel, ``{Deep Variational Privacy Funnel: General Modeling with Applications in Face Recognition},'' \emph{arXiv preprint arXiv:2401.14792}, 2024.

\bibitem{huang2022splitting}
T.-H. Huang, A.~El~Gamal, and H.~El~Gamal, ``{A Linearly Convergent Douglas-Rachford Splitting Solver for Markovian Information-Theoretic Optimization Problems},'' \emph{IEEE Transactions on Information Theory}, vol.~69, no.~5, pp. 3372--3399, 2022.

\bibitem{huang2024dca}
T.-H. Huang and H.~E. Gamal, ``{An Efficient Difference-of-Convex Solver for Privacy Funnel},'' \emph{arXiv preprint arXiv:2403.04778}, 2024.

\bibitem{gupta2011EM}
M.~R. Gupta and Y.~Chen, ``{Theory And Use of the EM Algorithm},'' \emph{Foundations and Trends{\textregistered} in Signal Processing}, vol.~4, no.~3, pp. 223--296, 2011.

\bibitem{kolchinsky2019nonlinear}
A.~Kolchinsky, B.~D. Tracey, and D.~H. Wolpert, ``{Nonlinear Information Bottleneck},'' \emph{Entropy}, vol.~21, no.~12, p. 1181, 2019.

\bibitem{ye2022lmrate}
W.~Ye, H.~Wu, S.~Wu, Y.~Wang, W.~Zhang, H.~Wu, and B.~Bai, ``{An Optimal Transport Approach to the Computation of the LM Rate},'' \emph{IEEE Global Communications Conference (GLOBECOM)}, pp. 239--244, 2022.

\bibitem{wu2023commot}
S.~Wu, W.~Ye, H.~Wu, H.~Wu, W.~Zhang, and B.~Bai, ``{A Communication Optimal Transport Approach to the Computation of Rate Distortion Functions},'' \emph{IEEE Information Theory Workshop (ITW)}, 2023.

\bibitem{blake1998uci}
C.~L. Blake, ``{UCI Repository of Machine Learning Databases},'' \emph{https://archive.ics.uci.edu}, 1998.

\end{thebibliography}

\newpage
\begin{appendices}
\section{Convergence Analysis}
We estimate the descent caused by the update of each variable in two adjacent iterations. 

\subsubsection{Descent caused by the update of $\bdq$}
The update rule of $\bdq$ gives
\begin{align}
&\tilde{f}(\bdu^{(n)},\bdr^{(n)},\bdq^{(n)})-\tilde{f}(\bdu^{(n)},\bdr^{(n)},\bdq^{(n+1)}) \notag \\
=&\sum\limits_{i,j,k}s_{ki}u_{ij}^{(n)}\log\dfrac{q_{ijk}^{(n+1)}}{q_{ijk}^{(n)}} \notag \\
=&\sum\limits_{j,k}\Big(\sum\limits_i s_{ki}u_{ij}^{(n)}\Big)\sum\limits_i q_{ijk}^{(n+1)}\log\dfrac{q_{ijk}^{(n+1)}}{q_{ijk}^{(n)}} \label{update-of-q} \\
=&\sum\limits_{j,k}\Big(\sum\limits_i s_{ki}u_{ij}^{(n)}\Big)D(\bdq_{jk}^{(n+1)}\Vert\bdq_{jk}^{(n)})\geq 0. \notag
\end{align}
In \eqref{update-of-q} we apply the update rule of $\bdq$, then we change the summation order to sum over $i$ first. 

\subsubsection{Descent caused by the update of $\bdu$}
The update rule of $\bdu$ gives
\begin{align}
&\tilde{f}(\bdu^{(n)},\bdr^{(n)},\bdq^{(n+1)})-\tilde{f}(\bdu^{(n+1)},\bdr^{(n)},\bdq^{(n+1)}) \notag \\
=&\sum\limits_{i,j,k}u_{ij}^{(n)}\Big(s_{ki}\log\dfrac{s_{ki}u_{ij}^{(n)}}{q_{ijk}^{(n+1)}r_j^{(n)}}-\lambda^{(n+1)}\log w_{ij}^{(n-1)}\Big) \notag \\
&-\sum\limits_{i,j,k}u_{ij}^{(n+1)}\Big(s_{ki}\log\dfrac{s_{ki} u_{ij}^{(n+1)}}{q_{ijk}^{(n+1)}r_j^{(n)}}-\lambda^{(n+1)}\log w_{ij}^{(n)}\Big) \label{adding-lagrangian-term} \\
=&\sum\limits_{i,j}u_{ij}^{(n)}\log\dfrac{u_{ij}^{(n)}}{e^{\lambda^{(n+1)}\log w_{ij}^{(n)}+\phi_{ij}^{(n+1)}}r_j^{(n)}} \notag \\
&-\sum\limits_{i,j}u_{ij}^{(n+1)}\log\dfrac{p_i}{\sum\limits_{j'}e^{\lambda^{(n+1)}\log w_{ij'}^{(n)}+\phi_{ij'}^{(n+1)}}r_{j'}^{(n)}} \notag \\
&+\lambda^{(n+1)}\sum\limits_{i,j}u_{ij}^{(n)}\log\dfrac{w_{ij}^{(n)}}{w_{ij}^{(n-1)}} \label{update-of-u} \\
=&\sum\limits_{i,j}u_{ij}^{(n)}\log\dfrac{u_{ij}^{(n)}}{e^{\lambda^{(n+1)}\log w_{ij}^{(n)}+\phi_{ij}^{(n+1)}}r_j^{(n)}} \notag\\
&-\sum\limits_{i,j}u_{ij}^{(n)}\log\dfrac{p_i}{\sum\limits_{j'}e^{\lambda^{(n+1)}\log w_{ij'}^{(n)}+\phi_{ij'}^{(n+1)}}r_{j'}^{(n)}} \notag\\
&+\lambda^{(n+1)}\sum\limits_{i,j}u_{ij}^{(n)}\log\dfrac{w_{ij}^{(n)}}{w_{ij}^{(n-1)}} \label{change-index} \\
=&\sum\limits_{i,j}u_{ij}^{(n)}\log\dfrac{u_{ij}^{(n)}}{u_{ij}^{(n+1)}}+\lambda^{(n+1)}\sum\limits_{i,j}r_j^{(n)}w_{ij}^{(n)}\log\dfrac{w_{ij}^{(n)}}{w_{ij}^{(n-1)}} \notag\\
=&D(\bdu^{(n)}\Vert\bdu^{(n+1)})+\lambda^{(n+1)}\sum\limits_j r_j^{(n)}D(\bdw_j^{(n)}\Vert\bdw_j^{(n-1)})\geq 0.\notag
\end{align}

%
Noticing that the update rule of $\bdu$ yields
\begin{equation*}
\hat{R}=\sum\limits_{i,j}u_{ij}^{(n+1)}\log w_{ij}^{(n)}=\sum\limits_{i,j}u_{ij}^{(n)}\log w_{ij}^{(n-1)}
\end{equation*}
in two consecutive iterations, we add these terms to get \eqref{adding-lagrangian-term} such that the update rule of $\bdu$ can be applied in \eqref{update-of-u}. 
The marginal distribution $p_i=\sum\limits_j u_{ij}^{(n)}=\sum\limits_j u_{ij}^{(n+1)}$ implies the derivation of \eqref{change-index}. 
In the final representation, $\bdw_j$ denotes the $j$-th column of $\bdw$. 

\subsubsection{Descent caused by the update of $\bdr$}
The update rule of $\bdr$ yields
\begin{align}
&\tilde{f}(\bdu^{(n+1)},\bdr^{(n)},\bdq^{(n+1)})-\tilde{f}(\bdu^{(n+1)},\bdr^{(n+1)},\bdq^{(n+1)}) \notag\\
=&\sum\limits_{j,k}\Big(\sum\limits_i s_{ki}u_{ij}^{(n+1)}\Big)\log\dfrac{r_j^{(n+1)}}{r_j^{(n)}} \notag\\
=&\sum\limits_{i,j}u_{ij}^{(n+1)}\log\dfrac{r_j^{(n+1)}}{r_j^{(n)}}=\sum\limits_j r_j^{(n+1)}\log\dfrac{r_j^{(n+1)}}{r_j^{(n)}} \label{update-of-r}\\
=&D(\bdr^{(n+1)}\Vert\bdr^{(n)})\geq 0.\notag
\end{align}
We change the summation order to sum over $k$ first, and in \eqref{update-of-r} we apply the update rule of $\bdr$. 

With the descent estimated, we can derive Lemma 1 in the context.

\section{Explanation in Large-Scale Experiment}
In \cite{huang2022splitting}, the objective is written as $\mathcal{L}(p,q,v)=F(p)+G(q)+\langle v,Ap-Bq\rangle+\dfrac{c}{2}\Vert Ap-Bq\Vert^2$, where $p,q$ represent $P_{Y|S}$ and $P_{Y|X}$ respectively, $A,B$ are coefficients. 
The function $G(q)$ is $\sigma_G$-weakly convex with $\sigma_G=2|\mathcal{Y}|M_q\Big(|\beta-1|+|\mathcal{X}|\Big)$, where $|\mathcal{X}|,|\mathcal{Y}|$ represent the cardinality of $\mathcal{X}$ and $\mathcal{Y}$ respectively, $\beta$ is the tradeoff parameter, and $M_q=\epsilon_{Y|X}^{-1}$, $\epsilon_{Y|X}$ is the infimum of $P_{Y|X}$. 
The subproblem
\begin{align*}
q^{k+1}=\underset{q\in\Omega_q}{\text{argmin}}~\mathcal{L}(p^{k+1},q,v^{k+1/2})
\end{align*}
is convex, and hence can be solved via gradient descent if
\begin{align*}
c&>M_q\dfrac{M_q\alpha\sigma_G+\sqrt{(M_q\alpha\sigma_G)^2+8(2-\alpha)L_q^2\lambda_B^2\mu_{BB^T}}}{4-2\alpha} \\
&>\dfrac{\alpha}{2-\alpha}M_q^2\sigma_G>\dfrac{2\alpha}{2-\alpha}M_q^3|\mathcal{X}||\mathcal{Y}|,
\end{align*}
where $\alpha$ is the relaxation coefficient, $\lambda_B$ and $\mu_{BB^T}$ are constant coefficients determined by $B$, and $L_q=\epsilon_{Z|X}^{-1}$, $\epsilon_{Z|X}^{-1}$ is the infimum of $P_{Z|X}$. 
The lower bound is proportional to $|\mathcal{X}|$ and $|\mathcal{Y}|$, and thus increases with the growing scale of the problem.  
Meanwhile, too large penalty brings numerical instability since it causes gradient explosion, consequently the numerical computation exceeds the computing ability of our device.  
Therefore, we can only obtain a locally sub-optimal solution by applying the DRS method to a large-scale dataset.

\end{appendices}

\end{document}